\documentclass[a4paper,11pt]{article}
\usepackage{amssymb,amsmath,amsthm}
\usepackage{texdraw}
\usepackage{a4wide}
\usepackage{authblk}

\newtheorem{theorem}{Theorem}[section]
\newtheorem{thm}{Theorem}[section]
\newtheorem{proposition}[thm]{Proposition}
\newtheorem{lemma}[thm]{Lemma}
\newtheorem{definition}[thm]{Definition}

\theoremstyle{remark}

\begin{document}

\title{3D compatible ternary systems and Yang-Baxter maps}

\author{Theodoros E. Kouloukas, Vassilios G. Papageorgiou}

\affil{\small{Department of Mathematics, University of Patras,
GR-265 00
Patras, Greece \\
tkoulou@master.math.upatras.gr, \ vassilis@math.upatras.gr}}

\date{}

 \maketitle

\begin{abstract}
According to Shibukawa, ternary systems defined on quasigroups and
satisfying certain conditions provide a way of constructing
dynamical Yang-Baxter maps. After noticing that these conditions can
be interpreted as 3-dimensional compatibility of equations on
quad-graphs, we investigate when the associated dynamical
Yang-Baxter maps are in fact parametric Yang-Baxter maps. In some
cases these maps can be obtained as reductions of higher dimensional
maps through compatible constraints. Conversely, parametric YB maps
on quasigroups with an invariance condition give rise to
3-dimensional compatible systems. The application of this method on
spaces with certain quasigroup structures provides new examples of
multi-parametric YB maps and 3-dimensional compatible systems.

\end{abstract}

\maketitle

\section{Introduction}

In \cite{shib,shib2}, Shibukawa studied the (set--theoretic)
dynamical Yang-Baxter (YB) equation on quasigroups and on loops. A
construction of dynamical YB maps was introduced in \cite{shib2},
from ternary systems defined on left quasigroups that satisfy
certain compatibility conditions. A parametric version of the latter
conditions turns to be equivalent with the 3-dimensional consistency
property of discrete (parametric) equations on lattices. Extra
`symmetry' conditions on the ternary system provide a way of
constructing parametric YB maps on groups and on particular
quasigroups. From the other hand, parametric YB maps satisfying an
invariance condition on quasigroups give rise to 3-dimensional
compatible systems. We apply these constructions on spaces with
certain quasigroup structures in order to derive (multi-)parametric
YB maps from 3-dimensional consistent equations on quad-graphs and
vice versa.

We begin in Section \ref{secshib} by giving the necessary
definitions for dynamical YB maps, parametric YB maps and Lax
matrices and by presenting the construction of dynamical YB maps due
to Shibukawa.

In Section \ref{sec3d}, we show the equivalence between the
3-dimensonal consistency property and the conditions of the previous
construction. By considering symmetry conditions on 3-dimensional
consistent equations on quasigroups we derive parametric YB maps. We
apply this method on 3-dimensional consistent equations on
quad-graphs. In each case, the corresponding quasigroup structure
that we consider is chosen according to the fulfilled symmetry
condition and gives rise to a different YB map.

In section \ref{secappl} we study the inverse problem in order to
derive 3-dimensional compatible systems from YB maps. We apply this
construction on the quasigroup $\mathbb{C}^{\ast}$ equipped with
division to a new four parametric quadrirational YB map on
$\mathbb{C}^{\ast}\times \mathbb{C}^{\ast}$. The latter map was
obtained from a higher dimensional map by imposing a compatible
constraint. The corresponding 3-dimensional compatible system
provides another four parametric map on the loop $\mathbb{C}$
equipped with substraction. We close this section by presenting a
3-dimensional compatible system on $GL_2(\mathbb{C})$ (that can be
generalized on $GL_n(\mathbb{C})$), derived by the general solution
of the re-factorization problem presented in \cite{kp,kp1}.

We conclude in Section \ref{secconcl} by giving some comments and
perspectives for future work.

\section{Quasigroups, Ternary systems and YB maps} \label{secshib}

 Let $H$, $X$ be two non-empty sets and $\phi$ a map from $H \times
X$ to $H$. A map
$$R(\lambda):X \times X
\rightarrow X \times X, ~~~\lambda \in H,$$ is called
\emph{dynamical Yang-Baxter  map}, with respect to $H$, $X$ and
$\phi$, if for any $\lambda \in H$ the map $R(\lambda)$ satisfies
the \emph{dynamical YB equation}

\begin{equation} \label{dYBeq}
R_{23}(\lambda)R_{13}(\phi(\lambda,X^{(2)}))R_{12}(\lambda)=
R_{12}(\phi(\lambda,X^{(3)})) R_{13}(\lambda)
R_{23}(\phi(\lambda,X^{(1)})).
\end{equation}
The maps $R_{12}(\lambda)$, $R_{12}(\phi(\lambda,X^{(3)}))$ etc, are
defined on the triple Cartesian product $X \times X \times X$, as
$$
R_{12}(\lambda)(u,v,w) = (R(\lambda)(u,v),w), \
R_{12}(\phi(\lambda,X^{(3)}))(u,v,w) =(R(\phi(\lambda,w))(u,v),w),
$$
for $u,v,w \in X$, and in a similar way $R_{13}$ and $R_{23}$. The
parameter $\lambda$ is called the dynamical parameter of the map
$R(\lambda)$ (and must not be confused with the parameters $\alpha,
~\beta$ of parametric YB maps). If $R(\lambda)$ does not depend on
$\lambda$, then it is a YB map, that is a map $R:X \times X
\rightarrow X \times X$ that satisfies the YB equation,
$R_{23}R_{13}R_{12}=R_{12}R_{13}R_{23}$, where by $R_{ij}$ for
$i,j=1,...,3$, we denote the action of the map $R$ on the $i$ and
$j$ factor of ${X} \times {X} \times {X}$.

A YB map $R:(X \times I) \times (X \times I) \rightarrow (X \times
I) \times (X \times I)$, with
\begin{equation} \label{pYB}
R:((x,\alpha),(y,\beta))\mapsto((u,\alpha),(v,\beta))= ((u(x,\alpha,
y,\beta),\alpha),(v(x,\alpha, y,\beta),\beta)),
\end{equation}
is called a \emph{parametric YB map} (\cite{ves4,ves2,ves3}). We
usually keep the parameters separately and denote a parametric YB
map as $R_{\alpha,\beta}(x,y):X\times X \rightarrow X \times X$.
From our point of view, the sets $X$ and $I$ have the structure of
an algebraic variety and the considered maps are birational. A {\em
Lax Matrix} of the YB map (\ref{pYB}) is a map $L:X\times I
\rightarrow Mat(n\times n)$,  that depends on a spectral parameter
$\zeta\in \mathbb{C}$, such that
\begin{equation} \label{laxmat}
L(u;\alpha)L(v;\beta)=L(y;\beta)L(x;\alpha).
\end{equation}
Furthermore, if equation (\ref{laxmat}) is equivalent to $(u,\
v)=R_{\alpha,\beta}(x,y)$ then $L(x;\alpha)$ will be called {\em
strong Lax matrix}.

\begin{definition}
A non empty set $L$ with a binary operation $\cdot :L \times L
\rightarrow L$ is called left quasigroup $(L,\cdot )$, if for every
$u, w \in L$ there is a unique $v \in L$ such that $u\cdot v=w$
(right quasigroups are defined in an analogous way).
\end{definition}
From this definition it turns out that in every left quasigroup one
additional operation is defined, namely  $\setminus :L \times L
\rightarrow L$, called \emph{left division}, with $u \setminus w =
v$ if and only if $u\cdot v=w$.

A left quasigroup $(L,\cdot )$ with the property stating that for
every $v, w \in L$ there is a unique $u \in L$ such that $u\cdot
v=w$, is called \emph{quasigroup}. In other words a quasigroup is a
left and right quasigroup. If there exist an element $e_l$ of a
quasigroup $L$ such that $e_l \cdot u=u$ for every $u \in L$, then
it is called \emph{left identity}, similarly $e_r \in L$ is called
right identity if $u \cdot e_r=u $, for every $u \in L$.
Furthermore, if there is an element $e \in L$ such that $u\cdot
e=e\cdot u =u$, for any $u \in G$, i.e. the left and right
identities coincide, then $(L,\cdot )$ is called \emph{loop}. We
usually omit the symbol $\cdot$ and write just $uv$ for $u \cdot v$.

Quasigroups and loops can be regarded as generalizations of groups.
In particular an associative loop is a group. For a detailed study
on quasigroups and loops we refer to \cite{pflu}.

Next, we present the construction of dynamical YB maps, introduced
by Shibukawa in \cite{shib2}. We begin with the following
definition.

\begin{definition}
A non-empty set $M$ equipped with a ternary operation $\mu:M \times
M \times M \rightarrow M$ is called a \emph{ternary system}
$(M,\mu).$
\end{definition}

Let $(L,\cdot )$ be a left quasigroup, $(M,\mu)$ a ternary system
and a bijective map $\pi:L \rightarrow M$. We consider the map
$R_\lambda: L \times L \rightarrow L \times L$, with
\begin{equation} \label{shYB}
R_\lambda(x,y)=(\eta_\lambda(y)(x),\xi_\lambda(x)(y)),
\end{equation}
where
\begin{align*}
\xi_\lambda(x)(y)&=\lambda \setminus \pi^{-1}
(\mu(\pi(\lambda),\pi(\lambda x), \pi((\lambda x)y))), \\
\eta_\lambda(x)(y)&=(\lambda \xi_\lambda(y)(x)) \setminus ((\lambda
y)x),  ~~~ \lambda,~x,~y \in L.
\end{align*}

\begin{theorem} $($Shibukawa$)$ \   \label{thshib}
The map $R_\lambda$ (\ref{shYB}) is a dynamical YB map with respect
to $L$, $L$ and $\phi:L \times L \rightarrow L$, with
$\phi(\lambda,x)=\lambda x$, if and only if
\begin{align}
\mu(a,\mu(a,b,c),\mu(\mu(a,b,c),c,d)) &= \mu(a,b,\mu(b,c,d)) ~~\text{και}  \nonumber\\
\mu(\mu(a,b,\mu(b,c,d)),\mu(b,c,d),d) &= \mu(\mu(a,b,c),c,d))
\label{mmm}
\end{align}
for every $a,b,c,d \in M$.
\end{theorem}
The dynamical YB map (\ref{shYB}) is a generalization of the YB map
that was presented in \cite{LYZ}, on any group that acts on itself
and the action satisfies a compatibility condition.
\section{3--Dimensional Consistency and YB Maps} \label{sec3d}
We are going to show that the conditions of theorem \ref{thshib} can
be interpreted as 3-dimensional consistency property of a quad-graph
equation.

We consider a parametric ternary operation $\mu_{\alpha,\beta}:X
\times X \times X \rightarrow X$, $\alpha,\beta \in I$, the
corresponding parametric equation
\begin{equation} \label{mcube}
w=\mu_{\alpha,\beta}(a,b,c),
\end{equation}
the initial values $a,b,c,d$ placed on the vertices of a cube and
the parameters $\alpha, \ \beta, \ \gamma$ assigned to the edges, as
shown in Fig.\ref{figcube2}. All parallel edges carry the same
parameter.

\begin{figure}[h]
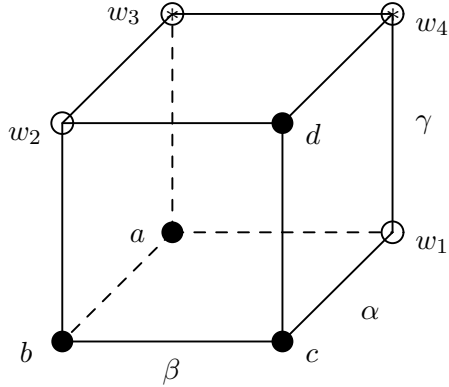

\bigskip

\centertexdraw{ \setunitscale 1.14 \linewd 0.01

\move (3 0) \linewd 0.01 \lpatt(0.05 0.05) \lvec (3.0 1.0)  \move (3
0) \lvec(2.5 -0.5)  \move (3 0) \lvec(4 0) \lpatt() \linewd 0.01
\lvec(4 1.0) \lvec(3 1) \lvec(2.5 0.5)  \lvec(2.5 -0.5) \lvec(3.5
-0.5) \lvec(4 0) \move(3.5 -0.5) \lvec(3.5 0.5) \lvec(4 1) \move(3.5
0.5) \lvec(2.5 0.5)

\move(2.5 -0.5) \fcir f:0.0 r:0.05 \move(3.5 -0.5) \fcir f:0.0
r:0.05 \move(4 0) \lcir r:0.05 \move(4 1) \lcir r:0.05 \move(3 0)
\fcir f:0.0 r:0.05 \move(3 1) \lcir r:0.05 \move(2.5 0.5) \lcir
r:0.05 \move(3.5 0.5) \fcir f:0.0 r:0.05

\htext (3.6 -0.6) {$c$} \htext (4.1 -0.1) {$w_{1}$} \htext (2.3
-0.6) {$b$} \htext (3.6 0.4) {$d$} \htext (3.85 -0.4) {$\alpha$}
\htext (2.95 -0.7) {$\beta$} \htext (4.1 0.45) {$\gamma$} \htext
(2.7 0.95) {$w_{3}$} \htext (2.8 -0.05) {$a$} \htext (4.1 0.9)
{$w_{4}$} \htext (2.25 0.4) {$w_2$} \htext (3.965 0.92){*} \htext
(2.961 0.92){*}

} \caption{Consistency around the cube } \label{figcube2}
\end{figure}

The values $w_{1}$ and $w_{2}$ are determined uniquely from the
initial conditions on the cube and equation (\ref{mcube}) as follows
$$w_{1}=\mu_{\alpha,\beta}(a,b,c), \ w_2=\mu_{\beta,\gamma}(b,c,d).$$
From the other hand, the values $w_3$ and $w_4$ can be derived by
two different ways. Following Fig.\ref{figcube2} we have that
\begin{align*}
 w_3 &=\mu_{\beta,\gamma}(a,w_1,\mu_{\alpha,\gamma}(w_1,c,d)) \ \text{or} \  w_3=\mu_{\alpha,\gamma}(a,b,w_2), \\
 w_4 &=\mu_{\alpha,\beta}(\mu_{\alpha,\gamma}(a,b,w_2),w_2,d) \
\text{or} \ w_4=\mu_{\alpha,\gamma}(w_1,c,d)).
\end{align*}
We will say that the equation (\ref{mcube}) satisfies the
\emph{3-dimensional (3D) consistency or compatibility condition} if
$w_3$ and $w_4$ are uniquely defined, that is the ternary operation
satisfies
\begin{align}
\mu_{\beta,\gamma}(a,\mu_{\alpha,\beta}(a,b,c),\mu_{\alpha,\gamma}(\mu_{\alpha,\beta}(a,b,c),c,d))
&=\mu_{\alpha,\gamma}(a,b,\mu_{\beta,\gamma}(b,c,d)), \label{3d1} \\
\mu_{\alpha,\beta}(\mu_{\alpha,\gamma}(a,b,\mu_{\beta,\gamma}(b,c,d)),\mu_{\beta,\gamma}(b,c,d),d)
&= \mu_{\alpha,\gamma}(\mu_{\alpha,\beta}(a,b,c),c,d)). \label{3d2}
\end{align}

This form of 3D compatibility can be traced back in \cite{adlery}
(see also \cite{pt}). An equivalent formulation was introduced both
in \cite{nij} (in order to derive Lax pair for the discrete
Krichever--Novikov equation) and in \cite{BS}, giving rise to the
ABS classification \cite{ABS1} of integrable $D_4$-symmetric
quad-graph equations.

Equations (\ref{3d1}) and (\ref{3d2}) are a parametric version of
the conditions of theorem \ref{thshib}. According to this
observation we give the next definition.

\begin{definition}
A non-empty set $X$ equipped with a parametric ternary operation
$\mu_{\alpha,\beta}:X \times X \times X \rightarrow X$ that
satisfies (\ref{3d1}) and (\ref{3d2}), is called 3D compatible
ternary system.
\end{definition}

So, according to theorem \ref{thshib}, every 3D compatible ternary
system on a quasigroup gives rise to a dynamical YB map. A natural
question is to  investigate instances where the dynamical YB maps
obtained in this way turn out to be independent of the dynamical
parametre $\lambda$ i.e. they are YB maps. In order to eliminate the
dynamical parameter $\lambda$ some extra symmetry conditions on
equations are helpful. In case of groups because of  associativity,
it is enough to consider homogeneous 3D consistent equations in the
following sense.

\begin{definition} \label{defhomog}
A 3D compatible ternary system on a quasigroup $(L,\cdot)$ is called
homogeneous if $\mu_{\alpha,\beta}(\lambda a,\lambda b,\lambda
c)=\lambda \mu_{\alpha,\beta}(a,b,c),$ for every $a,b, c,  \lambda
\in L$ and $\alpha,  \beta \in I $.
\end{definition}

The next proposition gives the parametric version of theorem
\ref{thshib}, for homogeneous 3D compatible ternary systems on
groups, that yield parametric YB maps.

\begin{proposition} \label{prsh1}
Let $(L,\mu_{\alpha,\beta})$ be a homogeneous 3D compatible ternary
system on the group $L$. The map
\begin{equation} \label{pYBq}
R_{\alpha,\beta}(x,y)=((\mu_{\alpha,\beta}(e,x,xy))^{-1}xy,\mu_{\alpha,\beta}(e,x,xy))
\end{equation}
is a parametric YB map.
\end{proposition}
The proof of this proposition is given below together with the proof
of proposition \ref{prQuasi}. The corresponding YB map for Abelian
groups with additive notation becomes
\begin{equation} \label{vshib}
R_{\alpha,\beta}(x,y)=((x+y-\mu_{\alpha,\beta}(0,x,x+y)),\mu_{\alpha,\beta}(0,x,x+y)).
\end{equation}

In the more general case of quasigroups the homogeneous condition is
not always sufficient. We are not going to investigate under which
algebraic conditions one can eliminate the dynamical parameter, but
instead we consider the following particular cases.

We consider the set $L=\mathbb{C}\setminus \{0 \}$ with the binary
operation $a\ast b= \frac{b}{a}$. This is a quasigroup with left
identity element $e_l=1$. In this quasigroup a symmetry condition of
3D compatible ternary systems that provides parametric YB maps is
given below.

\begin{proposition} \label{prQuasi}
Let $(L,\mu_{\alpha,\beta})$ be a 3D compatible ternary system on
the quasigroup $(L,*)$, such that
\begin{equation} \label{symQuasi}
\lambda
\mu_{\alpha,\beta}(a,b,c)=\mu_{\alpha,\beta}(\frac{a}{\lambda},\lambda
b,\frac{c}{\lambda}), \ \ \text{for every} \ \lambda \in L \
\text{and} \ \alpha, \beta \in I,
\end{equation}
then the map
\begin{equation} \label{pYBq2}
R_{\alpha,\beta}(x,y)=(\frac{y}{x}\mu_{\alpha,\beta}(1,x,\frac{y}{x}),\mu_{\alpha,\beta}(1,x,\frac{y}{x}))
\end{equation}
is a parametric YB map.
\end{proposition}

\begin{proof} We present the proof of propositions \ref{prsh1} and
\ref{prQuasi}. Let us denote by
$x',x'',y',y'',z',z'',\tilde{x},...,\tilde{\tilde{z}}$ the values
defined each time by the corresponding map (\ref{pYBq}) or
(\ref{pYBq2}) as follows,
\begin{align*}
R_{\alpha \beta}^{12}(x,y,z)=(x',y',z),  R_{\alpha, \gamma}^{13}
R_{\alpha\beta}^{12}(x,y,z)=(x'',y',z'), R_{\beta \gamma}^{23}
R_{\alpha,\gamma}^{13}
R_{\alpha \beta}^{12}(x,y,z)= (x'',y'',z''), \\
R_{\beta \gamma}^{23}(x,y,z)=(x,\tilde{y},\tilde{z}), \ \ R_{\alpha
\gamma}^{13} R_{\beta
\gamma}^{23}(x,y,z)=(\tilde{x},\tilde{y},\tilde{\tilde{z}}), \ \ \
R_{\alpha \beta}^{12} R_{\alpha,\gamma}^{13} R_{\beta
\gamma}^{23}(x,y,z)=
(\tilde{\tilde{x}},\tilde{\tilde{y}},\tilde{\tilde{z}}). \ \ \ \
\end{align*}
For convenience, we consider the map $T_{\alpha,\beta}:L\times L
\times L \rightarrow L\times L$, defined by
\begin{equation} \label{Tmap}
T_{\alpha,\beta}:(a, b, c) \mapsto
(\mu_{\alpha,\beta}(a,b,c)\setminus c,a\setminus
\mu_{\alpha,\beta}(a,b,c)). \end{equation} In Prop. \ref{prsh1},
where $L$ is a group $a\setminus b=a^{-1}b$, while for the
quasigroup of Prop. \ref{prQuasi} $a\setminus b=ab$, for every $a,b
\in L$. In both cases we can verify that
$R_{\alpha,\beta}(x,y)=T_{\alpha,\beta}(e_l,x,x\ast y),$ where $e_l$
is the corresponding left identity element (the identity element $e$
for Prop. \ref{prsh1} and $1$ for Prop. \ref{prQuasi}) and by $\ast$
we denote the corresponding binary quasigroup operation in both
cases. Next, we set
\begin{eqnarray*}
&& w_{1}:=\mu_{\alpha,\beta}(e_l,x,x \ast y), \
w_2:=\mu_{\beta,\gamma}(x,x \ast y,(x \ast y)\ast z),  \\
&& w_3:=\mu_{\beta,\gamma}(e_l,w_1,\mu_{\alpha,\gamma}(w_1,x\ast y,(x\ast y)\ast z))=\mu_{\alpha,\gamma}(e_l,x,w_2), \\
&& w_4:=\mu_{\alpha,\beta}(\mu_{\alpha,\gamma}(e_l,x,w_2),w_2,(x\ast
y)\ast z)=\mu_{\alpha,\gamma}(w_1,x\ast y,(x\ast y)\ast z)),
\end{eqnarray*}
where for $w_3$ and $w_4$ the last equality is obtained by equations
(\ref{3d1}) and (\ref{3d2}) respectively.

\begin{figure}[h]
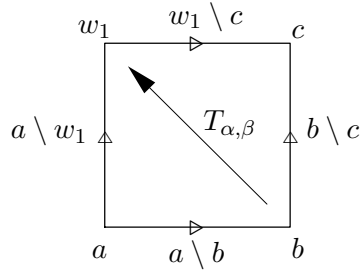

\centertexdraw{ \setunitscale 0.12 \linewd 0.01 \move (7 1) \linewd
0.03 \arrowheadtype t:F \avec(1 7) \lpatt( ) \move (0 0) \linewd
0.05 \lvec(8 0)  \lvec (8 8)  \lvec (0 8)  \lvec (0 0) \lpatt( )
\htext (3.7 -0.4){$\triangleright$}\htext (-0.5 -1.4){$a$} \htext
(2.8 -1.9){$a\setminus b$} \htext(8.1 -1.4){$ b$} \htext(8.1 8.2){$
c$} \htext(-1.1 8.2){$ w_1$}\htext (7.7
3.5){\tiny{$\bigtriangleup$}} \htext (8.75 3.5){$b\setminus c$}
\htext (3.7 7.6){$\triangleright$} \htext (2.8 8.5){$w_1\setminus
c$} \htext (-0.31 3.5){\tiny{$\bigtriangleup$}}\htext (-4 3.5){$
a\setminus w_1$} \htext (4.3 3.8){$T_{\alpha,\beta}$}} \caption{The
map $T_{\alpha,\beta}$ assigned to the edges of a quadrilateral}
\label{quadril}
\end{figure}

\begin{figure}[h]
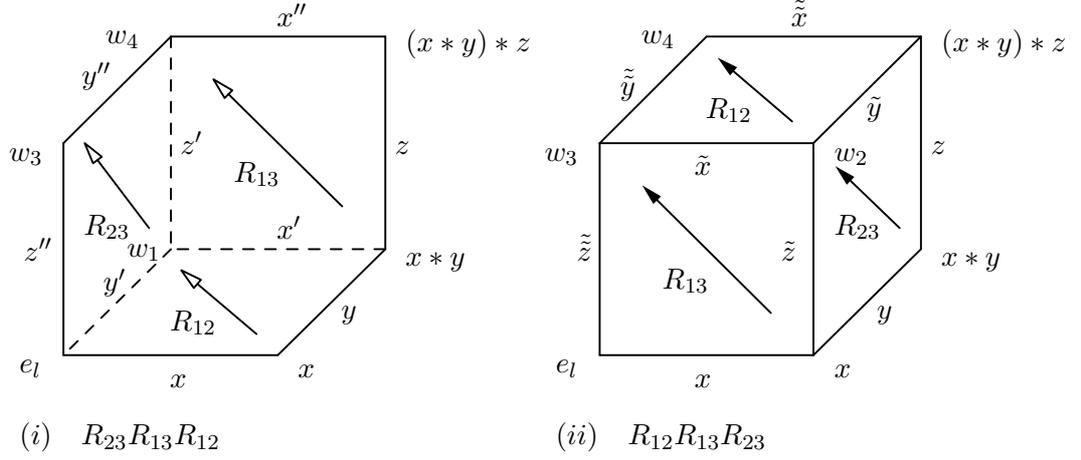

\centertexdraw{ \setunitscale 1.11 \linewd 0.06

\arrowheadsize l:0.1 w:0.05

\move (1. 0)  \linewd 0.01 \lvec(0.5 -0.5) \lpatt() \lvec(0.5
0.5)\lpatt() \lvec(1.0 1.0) \lpatt() \move(0.5 0.5)
\lpatt()\lvec(-0.5  0.5)\lpatt() \lvec(0.0 1.0)\lpatt(0.05 0.05)

\move(-2.5 0.0) \lpatt(0.05 0.05) \lvec(-3.0 -0.5) \lpatt()
\lvec(-2.0 -0.5) \lpatt() \lvec(-1.5 0.0) \lpatt(0.05 0.05)
\lvec(-2.5 -0.0)  \lpatt() \move(-1.5 0.0) \lvec(-1.5 1.0)\lpatt()
\lvec(-2.5 1.0) \lpatt(0.05 0.05) \lvec(-2.5 0.0) \lpatt()
\move(-3.0 -0.5) \lvec(-3.0 0.5) \lpatt()

\lvec(-2.5 1.0) \lpatt() \move(0 1) \lvec(1 1) \lpatt()\lvec(1
0)\lpatt() \move(-0.5 0.5) \lvec(-0.5 -0.5) \lpatt() \lvec(0.5
-0.5)\lpatt()

\move(-2.6 0.1) \avec(-2.9 0.5)\lpatt() \move(-1.7 0.2) \avec(-2.3
0.8)\lpatt() \move(-2.1 -0.4) \avec(-2.45 -0.1)\lpatt()
\arrowheadtype t:F \move(0.9 0.1) \avec(0.60 0.39)\lpatt() \move(0.3
-0.3) \avec(-0.3 0.3)\lpatt() \move(0.4 0.6) \avec(0.05 0.9)\lpatt()

\htext (-0.05 -0.65) {$x$} \htext (0.8 -0.35) {$y$} \htext (1.05
0.45) {$z$} \htext (0.35 -0.05) {$\tilde{z}$} \htext (0.75 0.6)
{$\tilde{y}$} \htext (-0.05 0.35) {$\tilde{x}$} \htext (-0.6 -0.05)
{$\tilde{\tilde{z}}$} \htext (-0.4 0.7) {$\tilde{\tilde{y}}$} \htext
(0.4 1.05) {$\tilde{\tilde{x}}$}

\htext (-2.5 -0.65) {$x$} \htext (-1.7 -0.35) {$y$} \htext (-1.45
0.45) {$z$} \htext (-2.0 0.05) {$x'$} \htext (-2.81 -0.22) {$y'$}
\htext (-2.0 1.05) {$x''$} \htext (-2.45 0.45) {$z'$} \htext (-2.91
0.75) {$y''$} \htext (-3.18 -0.05) {$z''$}

\htext (0.6 0.05) {$R_{23}$} \htext (-0.2 -0.2) {$R_{13}$} \htext
(0.0 0.6) {$R_{12}$}

\htext (-2.5 -0.4) {$R_{12}$} \htext (-2.2 0.3) {$R_{13}$} \htext
(-2.9 0.05) {$R_{23}$}

\htext (-3.2 -0.95)  {$(i) \ \ \ R_{23}R_{13}R_{12}$} \htext (-0.7
-0.95) {$(ii) \ \ \ R_{12}R_{13}R_{23}$}

\htext (0.6 -0.6) {$x$} \htext (1.1 -0.1) {$x\ast y$} \htext (-0.7
-0.6) {$e_l$} \htext (.6 0.4) {$w_2$}  \htext (-0.3 0.95) {$w_{4}$}
\htext (1.1 0.9) {$(x\ast y)\ast z$} \htext (-0.75 0.4) {$w_3$}

\htext (-1.9 -0.6) {$x$} \htext (-1.4 -0.1) {$x\ast y$} \htext (-3.2
-0.6) {$e_l$}  \htext (-2.8 0.95) {$w_{4}$} \htext (-2.7 -0.05)
{$w_1$} \htext (-1.4 0.9) {$(x\ast y)\ast z$} \htext (-3.25 0.4)
{$w_3$}

} \caption{Cubic representation of the Yang--Baxter property}
\label{CYB}
\end{figure}

Using the previous relations, the definition of $T_{\alpha,\beta}$
and the corresponding symmetry condition $\lambda
\mu_{\alpha,\beta}(a,b,c)=\mu_{\alpha,\beta}(\lambda a,\lambda
b,\lambda c)$ for proposition \ref{prsh1} and  $\lambda
\mu_{\alpha,\beta}(a,b,c)=\mu_{\alpha,\beta}(\frac{a}{\lambda},\lambda
b,\frac{c}{\lambda})$ for proposition \ref{prQuasi}, we can verify
that
\begin{align*}
(x',y')=T_{\alpha,\beta}(e_l,x,x\ast y), \
(x'',z')=T_{\alpha,\gamma}(w_1,x\ast y,(x\ast y)\ast z), \
(y'',z'')=T_{\beta,\gamma}(e_l,w_1,w_4), \\
(\tilde{y},\tilde{z})=T_{\beta,\gamma}(x,x\ast y,(x\ast y)\ast z), \
(\tilde{x},\tilde{\tilde{z}})=T_{\alpha,\gamma}(e_l,x,w_2), \
(\tilde{\tilde{x}},\tilde{\tilde{y}})=T_{\alpha,\beta}(w_3,w_2,(x\ast
y)\ast z),
\end{align*}
and from (\ref{Tmap}) we derive $x''=w_4 \setminus ((x\ast y)\ast
z)=\tilde{\tilde{x}}, \ y''=w_3 \setminus w_4=\tilde{\tilde{y}}$ and
$z''=w_3=\tilde{\tilde{z}}$.
\end{proof}
A similar construction on the loop $(\mathbb{C},*)$, with $a\ast
b=b-a$, satisfying an equivalent symmetry condition, can be applied
as well. In this case the corresponding symmetry condition of
proposition \ref{prQuasi} becomes
\begin{equation} \label{symLoop}
\lambda
+\mu_{\alpha,\beta}(a,b,c)=\mu_{\alpha,\beta}({a}-{\lambda},\lambda
+ b,{c}-{\lambda}),
\end{equation}
and the induced parametric YB map is
\begin{equation} \label{pYBLoop}
R_{\alpha,\beta}(x,y)=(\mu_{\alpha,\beta}(0,x,{y}-{x})+y-x,\mu_{\alpha,\beta}(0,x,{y}-{x})).
\end{equation}
The proof is similar with the proof of proposition \ref{prQuasi} by
considering the operation $a \ast b=b-a$, the left identity element
$e_l=0$ and the left division $a \setminus b= a+b$.

We summarize the latter results in the next proposition.

\begin{proposition}
Let $(L,\mu_{\alpha,\beta})$ be a 3D compatible ternary system on an
Abelian group $(L,\cdot )$, such that
\begin{equation} \label{symQuasiab}
\lambda \cdot \mu_{\alpha,\beta}(a,b,c)=\mu_{\alpha,\beta}(a \cdot
{\lambda}^{-1},\lambda \cdot b,{c} \cdot {\lambda}^{-1}), \ \
\text{for every} \ \lambda \in L \ \text{and} \ \alpha, \beta \in I,
\end{equation}
then the map
\begin{equation} \label{pYBq2ab}
R_{\alpha,\beta}(x,y)=({y}\cdot{x}^{-1} \cdot
\mu_{\alpha,\beta}(e,x,{y}\cdot{x}^{-1}),\mu_{\alpha,\beta}(e,x,{y}\cdot{x}^{-1}))
\end{equation}
is a parametric YB map.
\end{proposition}

\begin{proof}
On any group $(L,\cdot )$ we can define the binary operation $\ast:L
\times L \mapsto L$, $a \ast b=b \cdot a^{-1}$. It is easy to verify
that $(L,\ast)$ is a qausigroup with left identity element the
identity element $e$ of the group $(L,\cdot )$ and left division
$a\setminus b =a \cdot b$. The proof follows as in Prop.
\ref{prQuasi} for
\begin{equation*}
T_{\alpha,\beta}:(a, b, c) \mapsto (\mu_{\alpha,\beta}(a,b,c) \cdot
c,a \cdot \mu_{\alpha,\beta}(a,b,c)) \ \text{and} \
R_{\alpha,\beta}(x,y)=T_{\alpha,\beta}(e,x,x\ast y),
\end{equation*}
using the symmetry condition (\ref{symQuasiab}) and the fact that
$(L,\cdot )$ is Abelian.
\end{proof}

We have to notice that the previous propositions are inspired by the
symmetry method used in \cite{pstv2} in order to obtain YB maps out
of 3D-compatible quad-graph equations. We provide several examples
in the next section.

\subsection{YB maps from the ABS classification list}
We consider 3D compatible ternary systems from two integrable
equations on quad-graphs of the ABS classification list \cite{ABS1}.
In each case, the corresponding quasigroup structure that we
consider is chosen according to the symmetry condition that they
satisfy.

\subsubsection{YB maps from the $Q_1$ equation}
The equation
$$\alpha(a-w)(b-c)-\beta(a-b)(w-c)=0, \ \ a,b,c,w,\alpha,\beta \in \mathbb{C}$$
(the equation $Q_1$ of the $Q$ list of Adler, Bobenko, Suris
\cite{ABS1}) is 3D consistent. By solving it with respect to $w$, we
define the 3D compatible ternary system on $(\mathbb{C},+)$
\begin{equation} \label{Q1tern}
w=\mu_{\alpha,\beta}(a,b,c)=\frac{\alpha a(b-c)+\beta c(a-b)}{\alpha
(b-c)+\beta (a-b)},
\end{equation}
which is homogeneous, i.e. $\mu_{\alpha,\beta}(\lambda + a,\lambda+
b,\lambda+ c)= \lambda+ \mu_{\alpha,\beta}(a,b,c)$. The
corresponding YB map (\ref{vshib}) is
$$R_{\alpha,\beta}(x,y)=(\frac{\alpha y(x+y)}{\beta x+\alpha
y},\frac{\beta x(x+y)}{\beta x+\alpha y}).$$

Furthermore, equation (\ref{Q1tern}) defines a homogeneous 3D
compatible ternary system on $(\mathbb{C}^{\ast},\cdot)$, since
$\mu_{\alpha,\beta}(\lambda  a,\lambda b,\lambda c)= \lambda
\mu_{\alpha,\beta}(a,b,c)$. In this case the corresponding YB map of
Prop. \ref{prsh1} is
\begin{align*}
R_{\alpha,\beta}(x,y)&=(\frac{x y}{\mu_{\alpha,\beta}(1,x,x
y)},\mu_{\alpha,\beta}(1,x,x y)) \\
&= (y \frac{\alpha x y +(\beta-\alpha)x-\beta}{\beta x y
+(\alpha-\beta)y-\alpha},x \frac{\beta x y
+(\alpha-\beta)y-\alpha}{\alpha x y +(\beta-\alpha)x-\beta}),
\end{align*} that is the $H_{II}$ YB map of the $H$ list of
\cite{pstv}.

\subsubsection{YB maps from the dKdV}

The discrete Korteweg-de-Vries equation
$$(c-a)(b-w)=\alpha-\beta, \ \ a,b,c,w,\alpha,\beta \in \mathbb{C}$$
is a 3D consistent equation. If we solve it with respect to $w$, we
define the ternary operation $\mu_{\alpha,\beta}$ by
$$w=\mu_{\alpha,\beta}(a,b,c)=b-\frac{\alpha-\beta}{c-a},$$
that satisfies the equations (\ref{3d1}) and (\ref{3d2}). Also,
$\mu_{\alpha,\beta}(\lambda + a,\lambda+ b,\lambda+ c)= \lambda+
\mu_{\alpha,\beta}(a,b,c)$, The corresponding YB map of Prop.
\ref{prsh1} is
\begin{align*}
R_{\alpha,\beta}(x,y) =
((x+y-\mu_{\alpha,\beta}(0,x,x+y)),\mu_{\alpha,\beta}(0,x,x+y)) =
(y+ \frac{\alpha-\beta }{x+y}, x- \frac{\alpha-\beta }{x+y}),
\end{align*}
i.e. the Adler's map.

Furthermore, $\lambda
\mu_{\alpha,\beta}(a,b,c)=\mu_{\alpha,\beta}(\frac{a}{\lambda},\lambda
b,\frac{c}{\lambda})$. By considering the 3D compatible ternary
system $(L,\mu_{\alpha,\beta})$ on the quasigroup $(L,*)$ of Prop.
\ref{prQuasi}, we derive the parametric YB map
$$R_{\alpha,\beta}(x,y)=(\frac{y}{x}\mu_{\alpha,\beta}(1,x,\frac{y}{x}),\mu_{\alpha,\beta}(1,x,\frac{y}{x}))=
(y(1+\frac{\alpha-\beta}{x-y}),x(1+\frac{\alpha-\beta}{x-y})),
$$
which is the $F_{IV}$ map of the $F$ list of the classification in
\cite{ABS2}.

Finally, $\mu_{\alpha,\beta}$ satisfies also the symmetry condition
(\ref{symLoop}). If we consider the 3D compatible ternary system
$(\mathbb{C},\mu_{\alpha,\beta})$ on the Loop $(\mathbb{C},*)$, with
$a \ast b=b-a$, then from (\ref{pYBLoop}) we derive the YB map
$$R_{\alpha,\beta}(x,y) =
(y+ \frac{\alpha-\beta }{x-y}, x+ \frac{\alpha-\beta }{x-y}),
$$
that is the $F_{V}$ map of the $F$ list of the classification in
\cite{ABS2}.

\section{From YB maps to 3D Consistent quad-graph equations} \label{secappl}

In this section we study the inverse problem, that is to derive 3D
compatible ternary systems from YB maps. For dynamical YB maps on
quasigroups satisfying an invariance condition, the answer is given
by Shibukawa in \cite{shib2}.

\begin{theorem} $($Shibukawa$)$ \   \label{thshib2}
Let $ R_\lambda(x,y)=(\eta_\lambda(y)(x),\xi_\lambda(x)(y)) $ be a
dynamical YB map on a left quasigroup $(L, \cdot)$, satisfying the
invariance condition
\begin{equation} \label{invshibu}
(\lambda \xi_\lambda(x)(y))\eta_\lambda(y)(x)=(\lambda x)y,
\end{equation}
for every $\lambda,x,y \in L$. Then the ternary operation $\mu$ on
$L$ defined by
\begin{equation} \label{ternshibu}
\mu(a,b,c)=a \xi_a(a\setminus b)(b\setminus c)
\end{equation}
is 3D-compatible (i.e. satisfies  equations (\ref{mmm})).
\end{theorem}
Since every YB map is a dynamical YB map, we can use this theorem in
order to derive 3D compatible ternary systems from YB maps
satisfying a corresponding invariance condition. It is easy to
verify that all the parametric YB maps that we presented in the
previous section satisfy condition (\ref{invshibu}). In particular,
if we denote by $R_{\alpha,\beta}:(x,y)\mapsto(u(x,y),v(x,y))$ each
map on the corresponding quasigroup $(L,\ast)$, then the condition
(\ref{invshibu}) becomes
\begin{equation} \label{invCo}
v(x,y) \ast u(x,y)=x \ast y.
\end{equation}

\subsection{Multiparametric YB maps and consistency} \label{application}
Next, as application of the last construction, we present first a
four parametric YB map obtained by reduction through a compatible
constraint of a higher dimensional map introduced in \cite{kp1} and
then apply theorem \ref{thshib2} to obtain a four parametric 3D
compatible quad-graph equation. We will use the next lemma.

\begin{lemma} \label{lema}
Let $ R_{\alpha,\beta}:X^2\times X^2 \rightarrow X^2 \times X^2$ be
a parametric YB map with
$$R_{\alpha,\beta}(x_1,x_2,y_1,y_2)=(u_1(x_1,x_2,y_1,y_2),u_2(x_1,x_2,y_1,y_2),v_1(x_1,x_2,y_1,y_2),v_2(x_1,x_2,y_1,y_2))$$
(here the functions $u_i, \ v_i$, for $i=1,2$, depend on the
parameters $\alpha$, $\beta$) and $f_{\alpha}:X \rightarrow X$ a
parameter depending function, such that
\begin{align}
u_2(x_1,f_{\alpha}(x_1),y_1,f_{\beta}(y_1)) &=
f_{\alpha}(u_1(x_1,f_{\alpha}(x_1),y_1,f_{\beta}(y_1))), \label{u2f} \\
v_2(x_1,f_{\alpha}(x_1),y_1,f_{\beta}(y_1)) &=
f_{\beta}(v_1(x_1,f_{\alpha}(x_1),y_1,f_{\beta}(y_1))), \label{u1f}
\end{align}
then the map $\tilde{R}_{\alpha,\beta}:X\times X \rightarrow X
\times X$, defined by
$$\tilde{R}_{\alpha,\beta}(x,y)=(u_1(x,f_{\alpha}(x),y,f_{\beta}(y)),v_1(x,f_{\alpha}(x),y,f_{\beta}(y)))$$
is a parametric YB map. Furthermore, if $L(x_1,x_2,\alpha)$ is a Lax
matrix of $R_{\alpha,\beta}$, then
$\tilde{L}(x,\alpha)=L(x,f_{\alpha}(x),\alpha)$ is a Lax matrix of
$\tilde{R}_{\alpha,\beta}$.
\end{lemma}
Direct computations prove this lemma (appendix). Lemma \ref{lema}
can be generalized on higher dimensional YB maps $
R_{\alpha,\beta}:X^n\times X^n \rightarrow X^n \times X^n $,
$$R_{\alpha,\beta}:(x_1,...,x_n,y_1,...,y_n) \mapsto
(u_1,...,u_n,v_1,...,v_n),$$  with a compatible parametric function
$f_{\alpha}:X^{n-1}\rightarrow X$, such that if we replace $x_k$,
$y_k$ by $f_{\alpha}(x_1,...,x_{k-1},x_{k+1},...x_n)$,
$f_{\beta}(y_1,...,y_{k-1},y_{k+1},...y_n)$ for $k=1,...n,$
respectively, then $u_k \mapsto
f_{\alpha}(u_1,...,u_{k-1},u_{k+1},...,u_n)$ and $v_k \mapsto
f_{\beta}(v_1,...,v_{k-1},v_{k+1},...,v_n)$. In this way the new
$u_i, \ v_i$, $i=1,...,n, \ i\neq k$, give rise to a YB map on
$X^{n-1}\times X^{n-1}$.

Now, we consider the YB map
$$\mathcal{R}_{\alpha,\beta}:((x_1,x_2),(y_1,y_2)) \mapsto
((U_{11},U_{12}),(V_{11},V_{12})),$$ where $U_{ij}, \ V_{ij}$ denote
the corresponding $ij$ elements of the matrices
\begin{align*}
U &= (
\bar{L}(y_1,y_2;\beta)\bar{L}(x_1,x_2;\alpha)+\frac{1}{\alpha_1
\alpha_2} K_{\alpha}K_{\beta})
(\bar{L}(y_1,y_2;\beta)K_{\alpha}+K_{\beta}
\bar{L}(x_1,x_2;\alpha))^{-1}K_{\alpha} , \\ V &=
K_{\alpha}^{-1}(\bar{L}(y_1,y_2;\beta)K_{\alpha}+K_{\beta}\bar{L}(x_1,x_2;\alpha)-UK_{\beta}),
\ \ \text{for} \ \alpha=(\alpha_1,\alpha_2),\
\beta=(\beta_1,\beta_2), \\
K_{\alpha}&= \begin{pmatrix}
\alpha_{1} & 0 \\
0 & \alpha_2
\end{pmatrix}
 \ \ \text{and} \ \
\bar{L}(x_{1},x_{2};\alpha)=%
\left(
\begin{array}{cc}
 {x_1} & {x_2} \\
 \frac{ \alpha_1-{\alpha_2} {x_1}^2}{{\alpha_1} {x_2}} & -\frac{\alpha_2
   {x_1}}{{\alpha_1}}
\end{array}
\right).
\end{align*}
This map is ``\emph{Case I}" quadrirational symplectic YB map of the
classification of binomial Lax matrices presented in \cite{kp1}, for
$a_3=-1$ and $a_4=0$, with strong Lax matrix
$\mathcal{L}(x_1,x_2,\alpha)=\bar{L}(x_{1},x_{2};\alpha)-\zeta
K_{\alpha}$ and Poisson bracket
$$
\{x_1,x_2\}=-\alpha_1 x_2, \  \{y_1,y_2\}=-\beta_1 y_2, \
\{x_i,y_j\}=0 \ \text{for} \ i=1,2, \ (\alpha_1,\beta_1,x_2,y_2 \neq
0).$$ By setting $x_1=y_1=0$ we derive $U_{11}=V_{11}=0$, and the
resulting from lemma \ref{lema} map
\begin{equation} \label{YBmys}
{R}_{\alpha,\beta}(x,y)=
 (y \frac{\beta_1 x+ \alpha_2 y}{\alpha_1 x+\beta_2 y},x \frac{\beta_1
x+ \alpha_2 y}{\alpha_1 x+\beta_2 y}),
\end{equation}
is a parametric YB map with strong Lax matrix
\begin{equation} \label{laxmys}
{L}(x;\alpha)=\mathcal{L}(0,x;\alpha)=\left(
\begin{array}{cc}
 -\alpha_1 \zeta & {x} \\
 \frac{1} {x} &  -\alpha_2 \zeta
\end{array}
\right).
\end{equation}
Furthermore, this YB map satisfies the invariance condition
\begin{equation} \label{invmys}
\frac{u(x,y)}{v(x,y)}=\frac{y}{x}, \ \text{for} \ u(x,y)=y
\frac{\beta_1 x+ \alpha_2 y}{\alpha_1 x+\beta_2 y} \ \text{and} \
v(x,y)=x \frac{\beta_1 x+ \alpha_2 y}{\alpha_1 x+\beta_2 y}.
\end{equation}
If we consider the YB map (\ref{YBmys}) as a birational map on the
quasigroup $L=\mathbb{C}\setminus \{0 \}$ with the binary operation
$a\ast b= \frac{b}{a}$, then the invariance condition (\ref{invmys})
corresponds to (\ref{invshibu}) and by theorem \ref{thshib2} we
obtain the corresponding 3D compatible ternary system on $(L,\ast)$
with
\begin{equation} \label{3dmyst}
\mu_{\alpha,\beta}(a,b,c)=b \frac{\beta_1 a+ \alpha_2 c}{\alpha_1
a+\beta_2 c}, \ \ \alpha=(\alpha_1,\alpha_2), \
\beta=(\beta_1,\beta_2),
\end{equation}
 or
equivalently,  by setting $w=\mu_{\alpha,\beta}(a,b,c)$,
\begin{equation}   \label{4paramqgeq}
w(\alpha_1 a+\beta_2 c)-b(\beta_1 a+ \alpha_2 c)=0.
\end{equation}

From the other hand, this 3D compatible ternary system on $(L,\ast)$
satisfies the symmetry condition (\ref{symQuasi}), and the
corresponding  YB map of Prop. \ref{prQuasi}
$$
R_{\alpha,\beta}(x,y)=(\frac{y}{x}\mu_{\alpha,\beta}(1,x,\frac{y}{x}),\mu_{\alpha,\beta}(1,x,\frac{y}{x})),
$$
coincides with the YB map (\ref{YBmys}).

We can reduce the number of parameters of the YB map (\ref{YBmys})
by setting $\alpha_1=\beta_1=c$ or $\alpha_2=\beta_2=c$. with $c$
constant. Another interesting reduction of parameters is derived by
setting $\alpha_1=\alpha-r, \ \alpha_2=\alpha+r, \ \beta_1=\beta-r,
\ \beta_2=\beta+r$, with $r$ constant. Then the  YB map
(\ref{YBmys}) is transformed to
\begin{equation} \label{YBmys2}
\bar{R}_{\alpha,\beta}(x,y)=(y
\frac{(\beta-r)x+(\alpha+r)y}{(\alpha-r)x+(\beta+r)y},x\frac{(\beta-r)x+(\alpha+r)y}{(\alpha-r)x+(\beta+r)y}),
\end{equation}
with strong Lax matrix $L(x,\alpha-r,\alpha+r)$. Here $r$ is not a
YB parameter but just a free parameter. The corresponding 3D
compatible ternary system of theorem \ref{thshib2}, $$w((\alpha-r)
a+(\beta+r) c)-b((\beta-r) a+ (\alpha+r) c)=0,$$ has been introduced
in \cite{NQC} and is a ``homotopy'' of discrete MKdV and Toda
equations.

The YB map (\ref{YBmys}) (as well as any map derived by reducing the
parameters described above) is quadrirational and belongs to the
subclass $[2:2]$ of \cite{ABS2}, nevertheless it is not an
involution since $R_{\alpha,\beta}\circ R_{\alpha,\beta} \neq id$.
Therefore, it is not equivalent with any map from the $F$ and $H$
families in \cite{pstv}. We notice that the induced 3D compatible
system (\ref{3dmyst}) is not $D_4$-symmetric. However, it is
homogeneous with respect to multiplication and generates another one
four parametric YB map (involution this time) according to
proposition \ref{prsh1},
$$R_{\alpha,\beta}(x,y)=(\frac{x y}{\mu_{\alpha,\beta}(1,x,x y)},\mu_{\alpha,\beta}(1,x,x
y))= (y \frac{\alpha_1  +\beta_2 xy}{\beta_1  +\alpha_2 xy},x
\frac{\beta_1  +\alpha_2 xy}{\alpha_1  +\beta_2 xy}),$$ that
satisfies the invariant condition (\ref{invCo}).

\subsection{3D Consistency on $GL_2(\mathbb{C})$}

We conclude this section by applying theorem \ref{thshib2} on the
general YB map presented in \cite{kp1}. Let $X, \ Y $ be two generic
elements of $GL_2(\mathbb{C})$ and $K: \mathbb{C}^d \rightarrow
GL_2(\mathbb{C})$ a $d$-parametric family of commuting matrices. It
is enough to consider as $K$ one of the two cases of the
classification in \cite{kp1} (three cases in $GL_2(\mathbb{R})$
respectively). Then the map
$$\mathcal{R}_{\alpha,\beta}(X,Y)=(U_{\alpha,\beta}(X,Y),V_{\alpha,\beta}(X,Y)), \ \ \text{where}$$
$U_{\alpha,\beta}(X,Y)=
(f_{2}^{\alpha}(X)YX-f_{0}^{\alpha}(X)K_{\alpha}K_{\beta})
(f_{2}^{\alpha}(X)(Y
K_{\alpha}+K_{\beta}X)-f_{1}^{\alpha}(X)K_{\alpha}K_{\beta})^{-1}K_{\alpha},$
 \\
$V_{\alpha,\beta}(X,Y) = K_{\alpha}^{-1}(YK_{\alpha}+K_{\beta}X-U_{\alpha,\beta}(X,Y)K_{\beta})$ \\ \\
and $f_i^{\alpha}$, $i=0,1,2$ are defined by the coefficients of the
characteristic polynomial $$\det(X-\zeta K_{\alpha})=
f_{2}^{\alpha}(X) \zeta^{2}-f_{1}^{\alpha}(X) \zeta +
f_{0}^{\alpha}(X),$$ is a quadrirational parametric YB map that
satisfies the invariant conditions
\begin{eqnarray*}
(U_{\alpha,\beta}(X,Y)-\zeta K_{\alpha})(V_{\alpha,\beta}(X,Y)-\zeta
K_{\beta})=(Y-\zeta K_{\beta})(X-\zeta K_{\alpha}), \\
\ \ \ \ \ \ \ \ f_i^{\alpha}(U_{\alpha,\beta}(X,Y))=f_i^{\alpha}(X),
\ f_i^{\beta}(V_{\alpha,\beta}(X,Y))=f_i^{\beta}(Y),
\end{eqnarray*}
for $i=0,1,2$.

\begin{proposition}
The ternary operation
\begin{equation} \label{ternGL}
\mu_{\alpha,\beta}(A,B,C)=V_{\alpha,\beta}(BA^{-1},CB^{-1})A
\end{equation}
defines a 3D compatible ternary system on $GL_2(\mathbb{C})$.
\end{proposition}

\begin{proof}
On $GL_2(\mathbb{C})$ we define the binary operation $A \ast B=BA$.
So, $A \setminus B=BA^{-1}$. Now from the first invariant condition
of $\mathcal{R}_{\alpha,\beta}$ we have that $$V_{\alpha,\beta}(X,Y)
\ast
U_{\alpha,\beta}(X,Y)=U_{\alpha,\beta}(X,Y)V_{\alpha,\beta}(X,Y)=YX=X
\ast Y,$$ which is equivalent with (\ref{invshibu}). So, by theorem
\ref{thshib2}, we conclude that
$(GL_2(\mathbb{C}),\mu_{\alpha,\beta})$ with
$$
\mu_{\alpha,\beta}(A,B,C)=A\ast V_{\alpha,\beta}(A \setminus B, B
\setminus C)=V_{\alpha,\beta}(BA^{-1},CB^{-1})A
$$
is a 3D compatible ternary system.
\end{proof}
This 3D compatible ternary system can be generalized in
$GL_n(\mathbb{C})$ by considering the corresponding YB maps on
$GL_n(\mathbb{C}) \times GL_n(\mathbb{C})$ from the recursive
formulae presented in \cite{kp1}.

\section{Conclusions} \label{secconcl}
We showed that the conditions of the construction of dynamical
Yang-Baxter maps out of ternary systems introduced by Shibukawa are
equivalent with the 3D consistency property of equations on
quadrilaterals. Moreover, certain symmetry conditions on the 3D
compatible ternary systems drop the dynamical parameter and yield
(plain) YB maps. It is clear that the underlying quasigroup
structure of the evolution space is crucial in this construction and
must be compatible with the symmetry condition of the initial
ternary system. From the other hand, parametric YB maps satisfying
an invariance condition give rise to 3D compatible systems. In this
case the suitable quasigroup structure can be traced from the
invariance condition of the map. Nevertheless, other quasigroup laws
with additional symmetry conditions on 3D consistent equations will
lead to new parametric YB maps with invariance conditions and vice
versa.

Incidentally, we saw in section \ref{application} that YB maps can
be derived by considering some compatible constraints (Lemma
\ref{lema}) on higher dimensional Poisson YB maps. The question that
arises is how to find these constraints and also how to find new
compatible Poisson structures in order to study the integrability of
the corresponding transfer maps on lattices (see for example
\cite{pnc,ves2,kp2}).

Finally, the relation of the presented non-involution YB maps with
the maps classified in \cite{pstv}, as well as the question of
existence and significance of multiparameter extensions of known YB
maps and 3D compatible quad-graph equations, is an interesting issue
that deserves further investigation.

\section*{Acknowledgements}

The authors thank the anonymous referees for their useful comments.
The  research of VGP has been co-financed by the European Union
(European Social Fund - ESF) and Greek national funds through the
Operational Program "Education and Lifelong Learning" of the
National Strategic Reference Framework (NSRF) - Research Funding
Program: THALES - Investing in knowledge society through the
European Social Fund.

\section*{Appendix }
The proof of Lemma \ref{lema}.
\begin{proof}
Suppose that $
 \tilde
R^{23}_{\beta,\gamma} \tilde R^{13}_{\alpha,\gamma} \tilde
R^{12}_{\alpha,\beta}(x,y,z)=(x',y',z'), \  \tilde
R^{12}_{\alpha,\beta} \tilde R^{13}_{\alpha,\gamma} \tilde
R^{23}_{\beta,\gamma}(x,y,z)=(x'',y'',z'')$ and
\begin{eqnarray*}
R^{23}_{\beta,\gamma}
 R^{13}_{\alpha,\gamma}
R^{12}_{\alpha,\beta}(x_1,x_2,y_1,y_2,z_1,z_2)=(\bar{x}_1,\bar{x}_2,\bar{y}_1,\bar{y}_2,\bar{z}_1,\bar{z}_2),
\\  R^{12}_{\alpha,\beta}  R^{13}_{\alpha,\gamma}
R^{23}_{\beta,\gamma}(x_1,x_2,y_1,y_2,z_1,z_2)=(\bar{\bar{x}}_1,\bar{\bar{x}}_2,\bar{\bar{y}}_1,\bar{\bar{y}}_2,\bar{\bar{z}}_1,\bar{\bar{z}}_2).
\end{eqnarray*}
 Then
$x'=u_1(u_1(x,f_{\alpha}(x),y,f_{\beta}(y)),f_{\alpha}(u_1(x,f_{\alpha}(x),y,f_{\beta}(y))),z,f_{\gamma}(z))$,
while
\begin{eqnarray*}
x''=u_1 ( u_1(x,f_{\alpha}(x),v_1(y,f_{\beta}(y),z,f_{\gamma}(z)),f_{\gamma}(v_1(y,f_{\beta}(y),z,f_{\gamma}(z)))), \\
  f_{\alpha} (u_1(x,f_{\alpha}(x),v_1(y,f_{\beta}(y),z,f_{\gamma}(z)),f_{\gamma}(v_1(y,f_{\beta}(y),z,f_{\gamma}(z))))), \\
  u_1(y,f_{\beta}(y),z,f_{\gamma}(z)),f_{\beta}(u_1(y,f_{\beta}(y),z,f_{\gamma}(z)))).
  \ \ \ \ \ \ \ \ \ \ \ \ \ \ \ \ \ \ \ \
\end{eqnarray*}
If we set $x=x_1, \ y=y_1, \ z=z_1, \ f_{\alpha}(x_1)=x_2, \
f_{\beta}(y_1)=y_2, \ f_{\gamma}(z_1)=z_2$, then from
(\ref{u2f}),(\ref{u1f}) we have that
\\
$
 f_{\alpha}(u_1(x,f_{\alpha}(x),y,f_{\beta}(y)))=u_2(x_1,x_2,y_1,y_2)$,
 \
 $f_{\beta}(u_1(y,f_{\beta}(y),z,f_{\gamma}(z))))=u_2(y_1,y_2,z_1,z_2)$,
 \\
$f_{\gamma}(v_1(y,f_{\beta}(y),z,f_{\gamma}(z))))=v_2(y_1,y_2,z_1,z_2)$,
\\
 $ f_{\alpha}
 (u_1(x,f_{\alpha}(x),v_1(y,f_{\beta}(y),z,f_{\gamma}(z)),f_{\gamma}(v_1(y,f_{\beta}(y),z,f_{\gamma}(z)))))=$ \\$
  u_2(x_1,x_2,v_1(y_1,y_2,z_1,z_2),v_2(y_1,y_2,z_1,z_2))$,
 so $x_1'=\bar{x}_1$, $x_1''=\bar{\bar{x}}_1$ and since
 $\bar{x}_1=\bar{\bar{x}}_1$, we derive that $x_1'=x_1''$ or
 $x'=x''$. In a similar way we can show that $y'=y''$ and $z'=z''$.

Furthermore, since
\begin{eqnarray*}
 &&L(u_1(x_1,x_2,y_1,y_2),u_2(x_1,x_2,y_1,y_2),\alpha)L(v_1(x_1,x_2,y_1,y_2),v_2(x_1,x_2,y_1,y_2),\beta)=
 \\
 &&L(x_1,x_2,\alpha)L(y_1,y_2,\beta), \ \text{for every } x_i,  y_i \in
 X, \ i=1,2,
\end{eqnarray*}
then \\  {\small{${\small{
L(u_1(x,f_{\alpha}(x),y,f_{\beta}(y)),u_2(x,f_{\alpha}(x),y,f_{\beta}(y)),\alpha)
 L(v_1(x,f_{\alpha}(x),y,f_{\beta}(y)),v_2(x,f_{\alpha}(x),y,f_{\beta}(y)),\beta)}}
$}} \\
 $= L(x,f_{\alpha}(x),\alpha)L(y,f_{\beta}(y),\beta), \ \text{for
every } x,  y \in
 X$, and from (\ref{u2f}),(\ref{u1f}) we get that
\begin{eqnarray*}
&&L(u_1(x,f_{\alpha}(x),y,f_{\beta}(y)),f_{\alpha}(u_1(x,f_{\alpha}(x),y,f_{\beta}(y))))
\cdot
\\
&&
L(v_1(x,f_{\alpha}(x),y,f_{\beta}(y)),f_{\beta}(v_1(x,f_{\alpha}(x),y,f_{\beta}(y))))
 = L(x,f_{\alpha}(x),\alpha)L(y,f_{\beta}(y),\beta),
\end{eqnarray*}
which means that $\tilde{L}(x,\alpha)=L(x,f_{\alpha}(x),\alpha)$ is
a Lax matrix of the YB map $\tilde{R}_{\alpha,\beta}$.

\end{proof}

\end{document}